\documentclass[10pt]{article}
\usepackage[utf8]{inputenc}
\usepackage[margin=1in]{geometry}
\usepackage{natbib}
\usepackage{amsmath}
\usepackage{amssymb}
\usepackage{amsthm}
\usepackage{subfig}
\usepackage{graphicx}
\usepackage{fancyhdr}
\usepackage{hyperref}
\usepackage{tikz}
\usepackage{enumerate}

\DeclareMathOperator{\E}{\mathbb{E}}
\DeclareMathOperator{\ESS}{ESS}

\newcommand{\ind}{\text{I}}
\newcommand{\R}{\mathbb{R}}
\newcommand{\N}{\mathbb{N}}
\usepackage[noend,ruled]{algorithm2e}

\newtheorem{lemma}{Lemma}
\newtheorem{example}{Example}

\newcommand\fp[1]{{\boldsymbol{\{}}#1{\boldsymbol{\}}}}

\title{The chopthin algorithm for resampling}
\author{Axel Gandy \qquad F. Din-Houn Lau \\Department of Mathematics,
  Imperial College London}
\date{}
\begin{document}
\maketitle

\begin{abstract}
  Resampling is a standard step in particle filters and more generally
  sequential Monte Carlo methods. We present an algorithm, called
  chopthin, for resampling weighted particles.  In contrast to
  standard resampling methods the algorithm does not produce a set of
  equally weighted particles; instead it merely enforces an upper
  bound on the ratio between the weights.  Simulation studies show
  that the chopthin algorithm consistently outperforms standard
  resampling methods.  The algorithms chops up particles with large
  weight and thins out particles with low weight, hence its name.  It
  implicitly guarantees a lower bound on the effective sample size.
  The algorithm can be implemented  efficiently, making it
  practically useful.  We show that the expected computational effort
  is linear in the number of particles.  Implementations for
  C\texttt{++}, R (on CRAN), Python and  Matlab are
  available.
\end{abstract}


\section{Introduction}\label{sec:introduction}

Particle filters and more generally sequential Monte Carlo methods
have gained importance and widespread use \cite{smcm}. One of their
key steps is resampling, which is intended to prevent weight
degeneracy. Broadly speaking, resampling starts with a set of
particles $x_1,\dots,x_n$ with associated weights $w_1,\dots, w_n$ and
produces a new set of particles (a subset of the original set with
potentially duplicates) with less uneven weights (often equal
weights).

A commonly used resampling algorithm is multinomial sampling, which
selects a new set of particles by sampling $n$ times with replacement
from $x_1,\dots,x_n$ with probabilities proportional to $w_1,\dots,
w_n$.  Other resampling schemes have been proposed, for example
systematic resampling
\citep{whitley1994,carpenter99:ImprovedPartFilter}, stratified
resampling \citep{Kitagawa1996strat}, residual resampling
\cite{liu1998sequential} and branching resampling \citep[p.\
278]{fundamentals}.  All of these algorithms return a set of particles
with equal weights.

The general consensus seems to be that, whilst it is possible to
outperform multinomial resampling, the more advanced methods such as
residual, stratified and systematic resampling are comparable in terms
of their performance in particle filters
\citep{douc2005comparison,hol2006resampling}.

In this article we show that it is possible to improve the performance
of the resampling step significantly.  We do this by presenting a new
resampling method that consistently outperforms the aforementioned
methods.  

The new algorithm, called chopthin, ensures that the weights are not
too uneven by enforcing an upper bound, $\eta$, on the ratio between
the resulting weight.  Chopthin can outperform other methods because
it does not return particles with equal weights.

The chopthin algorithm enforces the upper bound, $\eta$, on the ratio
between the weights, as follows: Particles with large weights, above a
threshold $a$, are potentially ``chopped'', i.e.\ replicated with the
original weight spread among the replicates. Particles with small
weights, below the threshold $a$, are ``thinned'' by randomly deciding
whether they should be deleted or kept, adjusting the weights by the
selection probability to ensure unbiasedness. A similar approach to
the thinning part of chopthin is used in \cite{fearnhead03:markov}
where the optimality of such a resampling method is shown in a certain
sense.

Particle filters often only perform the resampling step if a criterion
of the unevenness of the weights, such as the effective sample size
(ESS), drops below a fixed threshold. This avoids resampling if the
weights are relatively even and thus reduces the noise being
introduced through the resampling. This results in measures of the
evenness of the particles such as the ESS to fluctuate over time.

In contrast to this, chopthin can be executed at every step of a
particle filter. This is because chopthin evens out the weights less
than existing schemes. It will not alter the weights much (or at all)
if they are already relatively even. Using it at every step leads
to less fluctuation in the unevenness of the weights over time. Figure
\ref{fig:illustresampl} (later in the paper) illustrates this in an
example by looking at the ESS over time.

Chopthin can be implemented efficiently. Indeed, we present one
version of chopthin, which can be implemented in expected constant
linear effort in the number of particles.  

We begin by presenting the generic chopthin algorithm in Section
\ref{sec:alg}.  In Section \ref{sec:linear} we present a version of
the algorithm that has expected linear effort and show in a simulation
that its effort is comparable to other standard resampling methods.
Simulation studies are conducted in Section \ref{sec:sim1} that compares
the chopthin algorithm to other resampling schemes within a particle
filter. The results show that our new algorithm consistently
outperforms the other resampling methods. In Section
\ref{sec:controlESS} we prove that the algorithm implicitly controls
the ESS.

Implementations of chopthin are available: as an R-package (chopthin
on CRAN), as a python package (on the python package index), as
C\texttt{++} code and as a Matlab extension file (homepage of the
first author).

\section{The Generic Algorithm}\label{sec:alg}
Before introducing the chopthin algorithm, we first present the
constraints that it satisfies. Denote the $n$ particle weights
before resampling as $w_1,\dots,w_n$ and let $\mathcal{G}$ be the
$\sigma$-field generated by $w_1,\dots,w_n$. Further, denote the $N$
weights after resampling as $\widetilde{w}_1,\dots,\widetilde{w}_N$. Let
$C^i$ be the number of replicates of particle $i$. We want chopthin to
satisfy the following:\
\begin{enumerate}[(i)]
\item $\displaystyle\E(C^i \widetilde{w}_i| \mathcal{G}) = w_i$ $\forall i$\label{item:unbiasedness}\hfill (Unbiasedness)
\item $\displaystyle\sum_{i=1}^nC^i=N$\label{item:target-count} \hfill (Target count)
\item $\displaystyle\sum_{i=1}^nw_i = \sum_{i=1}^N\widetilde{w}_i$\label{item:converse-weight} \hfill (Conserve weight)
\item $\displaystyle\frac{\widetilde{w}_i}{\widetilde{w}_j}\leq \eta $ $\forall i,j$\label{item:bound-ratio}\hfill (Bounded ratio)
\end{enumerate}
Property (\ref{item:unbiasedness}) is an unbiasedness condition
ensuring the expected total weight of the offspring of a particle is
equal to its original weight. Property (\ref{item:target-count})
ensures that exactly $N$ particles are returned after
chopthin. Typically, $N=n$ i.e.\ the number of particles is
conserved. Properties (\ref{item:unbiasedness}) and
(\ref{item:target-count}) are satisfied by other resampling methods
\citep{douc2005comparison}. Property (\ref{item:converse-weight})
ensures that the total sum of the weights before and after resampling
are equal. Property (\ref{item:unbiasedness}) and
(\ref{item:converse-weight}) ensure that any estimator based on the
normalised weights will be unbiased. Finally, property
(\ref{item:bound-ratio}) bounds the ratio of weights returned from
chopthin.

Algorithm \ref{alg:generic} is a generic version of chopthin. As input
it receives the  weights 
$(w_i)_{1:n}$ of $n$ particles, $\eta$, the desired upper
bound on the ratio between weights, and $N$, the number of particles
to be returned.

Every particle gets a (potentially) random number of descendants. For
a particle with weight $w$, the expected number of offspring from
chopthin will be $h_a^\eta(w)$, where
$h_a^{\eta}:[0,\infty)\to [0,\infty)$ is a given function which may
depend on $\eta$ and on a further threshold parameter $a$. To ensure
that $N$ particles are returned (in expectation), we need to find $a$
such that
\begin{equation}
\label{eq:defa}
  \sum_{i=1}^n h_a^\eta(w_i)=N.
\end{equation}

The mechanism that generates the descendants depends on the weight of
the particle as well as on the parameter $\eta$, which is specified
by the user, and the parameter $a$, which is determined by the algorithm.

The key steps of  Algorithm \ref{alg:generic} are:\
\begin{itemize}
\item[] \textbf{Find $a$} (Step \ref{algstep:a}):\ The parameter $a$
  will serve as a threshold parameter that determines which particles
  are ``thinned'' and which are ``chopped''.
\item[]\textbf{Thin} (Step \ref{algstep:thin}):\ Particles with
 weights below $a$ get ``thinned'', i.e.\ either have 1 offspring
 (with weight $a$) or 0 offspring. 
\item[]\textbf{Chop} (Step \ref{algstep:sys_resampling2}):\ Particles
  with weights above $a$ get ``chopped'', which means that they get
  subdivided into smaller pieces, dividing the total original weight.
\end{itemize}

\begin{algorithm}[tb]
\caption{Generic chopthin}\label{alg:generic}
\DontPrintSemicolon \KwIn{particle weights $(w_i)_{1:n}$; maximal
  weight ratio $\eta$; target number of particles $N$; function
  $h_a^\eta:[0,\infty)\to[0,\infty)$} \KwOut{ancestors $I\in \{1,\dots,n\}^N$, weights $\widetilde{w}\in [0,\infty)^N$} 

\nl Let $a$ be a solution
to $\sum_{i=1}^n h_a^\eta(w_i)=N$.\label{algstep:a} 

 Let $L=\left\{j: w_j < a  \right\}$ and $U=\left\{j: w_j \geq  a  \right\}$\;
Let $I=()$ and $\widetilde w =()$\;

\nl
Draw $u \sim U(0,1)$\tcp*{Thin}
\For{$i\in L$}{
  $u=u+h_a^\eta(w_i)$ \;
  \If{$u\geq 1$}{
    append $(i)$ to $I$ and $(a)$ to $\widetilde w$\;
    $u=u-1$\;
  }
}\label{algstep:thin}%

\nl Let $N_L=\text{length}(I)$, 
  $N_U=N-N_L-\sum_{i\in U}\left\lfloor h_a^\eta(w_{i}) \right\rfloor$\;
and 
  $
    \zeta := \frac{\sum_{i\in L}w_i-a N_L }{\sum_{i\in U}\fp{h_a^\eta(w_{i})} }
  $\label{algstep:xi}%

  \nl  Systematic$\left[\left\{\fp{h_a^\eta(w_{j})} ; j\in U \right\} ;N_U  \right]$; returns $m^{j}$ for $j\in U$.\; \label{algstep:sys_resampling2}
 
  \For{$i\in U$}{
    $c=\left\lfloor h_a^\eta(w_{i})\right\rfloor + m^i$\;
    append $(i,\dots,i)\in \N^{c}$ to $I$ \;
     $\widehat{w}_i = w_i + \zeta\fp{h_a^\eta(w_{i})}$\tcp*{Adjusted weight}
      append $(\frac{\widehat{w}_i}{c},\dots,\frac{\widehat{w}_i}{c})\in \R^{c}$ to $\widetilde w$ \tcp*{Chop}
    } 

 \Return{$I,\widetilde w$}\;
\end{algorithm}
The chopthin algorithm returns a vector of resampled weights
$(\widetilde{w}_i)_{1:N}$ and an integer vector, $I$, containing the
indices of resampled components of the original weights. Chopthin will
return weights between $a$ and $\eta a$. This way the bound on the
ratio of the weights, property (\ref{item:bound-ratio}), will be
satisfied.  

We now discuss Algorithm \ref{alg:generic} in detail.
In Step \ref{algstep:a} the threshold
parameter $a$ is found by solving \eqref{eq:defa}. This depends on the
choice of function $h_a^\eta$. Choosing $h_a^\eta$ and solving
\eqref{eq:defa} are discussed toward to end of this section and in
Section \ref{sec:linear}.

The thinning step (Step \ref{algstep:thin}) 
determines the new weight and the number of offspring for particles
with small weights, $w_i<a$. Descending particles will have weight
$a$, thus ensuring the range condition on the weights. The
unbiasedness property (\ref{item:unbiasedness}) requires
$h_a^\eta(w)=w/a$, uniquely determining $h_a^\eta$ in this range. The
number of offspring is determined by systematic resampling on
$\{h_a^\eta(w_i):w_i<a \}$, ensuring $0$ or $1$ descendants. 

Step \ref{algstep:thin} returns $N_L$ particles such that
$\E(N_L|\mathcal{G})=\sum_{i:w_i< a} w_i/a$.  The total weight of the
surviving thinned particles is $aN_L$. Thus, through the thinning
step, the total sum of the weights may have changed. We compensate for
this using $\zeta$ (step \ref{algstep:xi}) in the chopping step, thus
ensuring property (\ref{item:converse-weight})

The chopping step (Step \ref{algstep:sys_resampling2}) determines how
the large weights, $w_i\geq a$, are subdivided. Each large weight will
receive $c^i=\left\lfloor h_a^\eta(w_i)\right\rfloor+m^i$
offspring. The $m^i$ are determined by a second systematic resampling
step on the fractional parts $\fp{h_a^\eta(w_i)}$ where $\fp{x} := x -
\left\lfloor x \right\rfloor$. Performing systematic resampling on
these fractional parts ensures
$\E(C^i|\mathcal{G})=h_a^\eta(w_i)$. This holds because the expected
value of $m^i$ is $(\zeta/a +1)\fp{h_a^\eta(w_i)}$ and
$\E(\zeta|\mathcal{G})=0$.  Further, this resampling step will return
exactly $N_U$ particles. The value of $N_U$ is selected such that the
total number of offspring produced from the entire algorithm is
exactly $N$ (step \ref{algstep:xi}). Thus property
(\ref{item:target-count}) is satisfied. Before chopping, the original
weight is first adjusted using $\zeta$. The adjusted weight is
$\widehat{w}_i=w_i+\zeta\fp{h_a^\eta(w_i)}$. This adjustment ensures
that the totals sum of the weights is conserved, property
(\ref{item:converse-weight}) and that the chopped weights are unbiased
(\ref{item:unbiasedness}).

The restriction that the chopped weights are between $a$ and
$\eta a$ requires
$$
a\leq \frac{\widehat{w}}{\lceil c \rceil} \text{ and }
\frac{\widehat{w}}{\lfloor c \rfloor}\leq \eta a,\quad\forall w\geq a
$$
where $\widehat{w}=w-\zeta\fp{h_a^\eta(w)}$ is the adjusted original
weight and $c$ is the number of offspring. These constraints
define an area, $A$, in the (adjusted) weight-count space
where
\begin{equation*}
  A:=\left\{(\widehat{w},c):   \left\lceil \frac{\widehat{w}}{\eta a}\right\rceil \leq c \leq   \left\lfloor\frac{\widehat{w}}{a}\right\rfloor, w\geq a\right\}.
\end{equation*}
This region is illustrated in Figure \ref{fig:chop_regions} by the
light grey area with black border for the case $\eta=4.5$. This area
is only valid if $\eta\geq 2$.
\begin{figure*}[tb]
  \centering
  \includegraphics[width=0.75\textwidth]{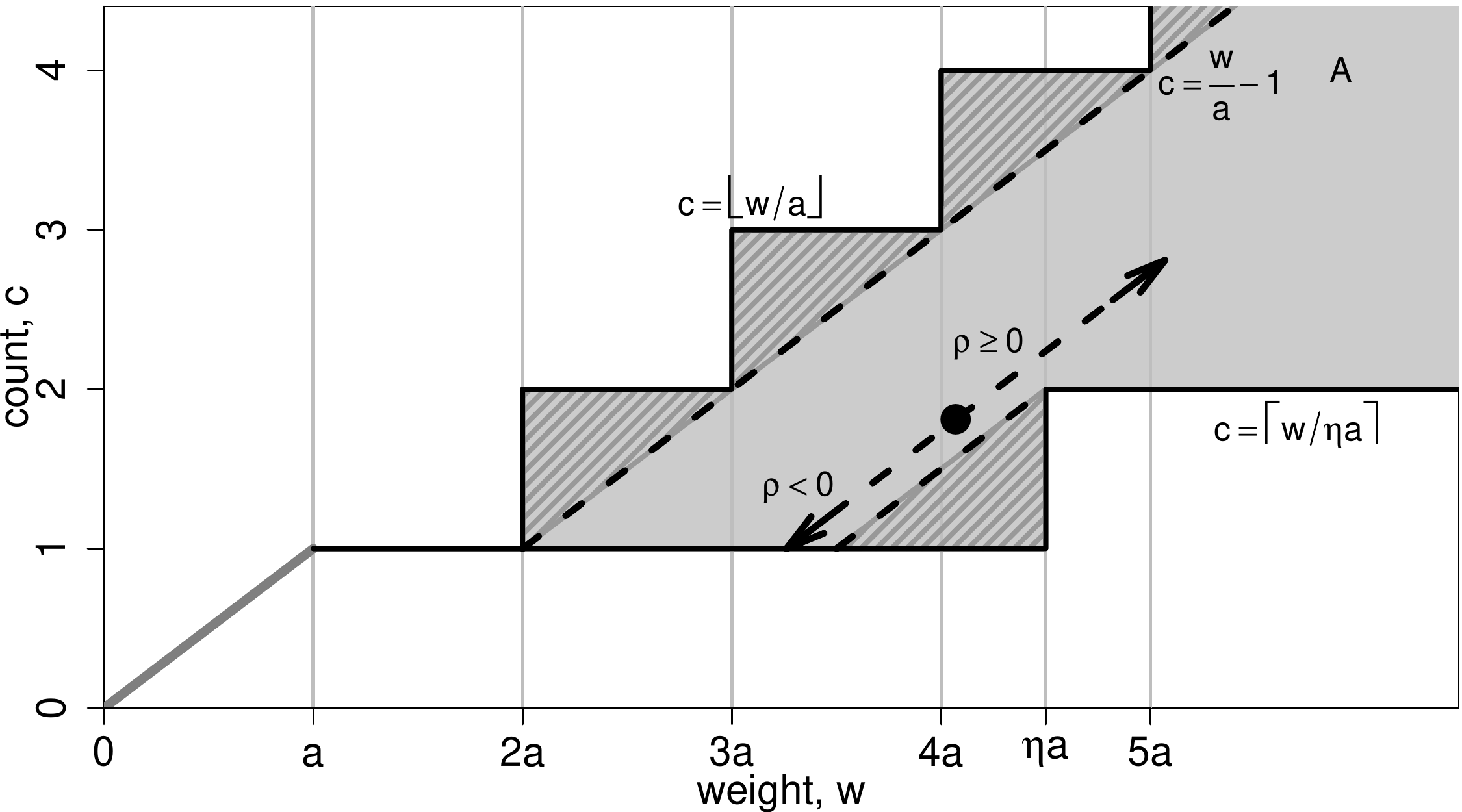}
  \caption{Illustration of allowable chopping region, $A$, represented
    by gray area with black border. The black point represents a
    candidate point $(w,h_a^\eta(w))$, the arrows represent movement
    from the point by the adjustment vector $\rho (1,1/a)$. The dark gray
    lined areas represent regions where $(w,h_a^\eta(w))$ cannot
    lie. All dashed lines have a gradient of
    $1/a$.}\label{fig:chop_regions}
\end{figure*}

To show that the $(\widehat{w},c)$ used in Algorithm \ref{alg:generic}
lies in $A$, we first write the adjusted weight and count as
\begin{equation}\label{eq:5}
  \left(\begin{array}{c}
    \widehat{w}\\
    c
  \end{array}\right):=\left(\begin{array}{c}
    w \\
    h_a^\eta(w)
  \end{array}\right)+\rho\left(\begin{array}{c}
    1\\
    1/a
  \end{array}\right),
\end{equation}
where $\rho:=\zeta\fp{h_a^\eta(w)}$. We shall refer to the
vector $\rho (1,1/a)$ as the \textit{adjustment} vector as it adjusts
the original weight, $w$, to the adjusted weight, $\widehat{w}$.

The requirement that $(\widehat{w},c)\in A$ leads to constraints on
$h_a^\eta(w)$. Beside choosing $h_a^\eta(w)$ such that
$(w,h_a^\eta(w))\in A$ we also need to ensure that the adjusted weight
$(\widehat{w},c)$ is also in $A$. Possible constraints ensuring this are
\begin{equation*}
  h_a^\eta(w)\leq \frac{w}{a}-1
\end{equation*}
and 
\begin{equation}\label{eq:1}
  h_a^\eta(w)\geq \frac{w}{a} - m(\eta-1)+1
\end{equation}
for $a(\eta m-1)<w<\eta am$, $m=1,2,\dots$. The regions where
$(w,h_a^\eta(w))$ are not allowed are represented by the dark grey
lined areas in Figure \ref{fig:chop_regions}. An alternative choice
for $h_a^\eta$ could be to choose such that
$\fp{h_a^\eta(w)}=0$ for all $w$ so that $\rho=0$ (see end
of this section). In this case, it is sufficient that
$(w,h_a^\eta(w))\in A$.

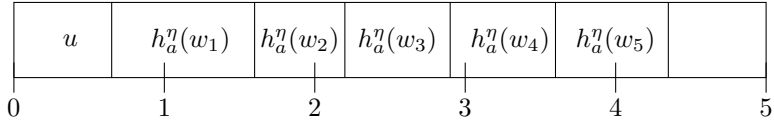
\begin{figure*}[btp]\centering  
  \begin{tikzpicture}[xscale=2]
   \draw (0,0) -- (5,0) -- (5,1) -- (0,1) -- (0,0);
   \foreach \x in {0,1,2,3,4,5} { 
         \draw(\x,-0.2)--(\x,0.2);
         \draw (\x, -0.4) node {\x};
    }
   \foreach \x in {0.65,1.6,2.2,2.9,3.6,4.35}{
          \draw (\x,0) --(\x,1);
   }
   \draw (0.375,0.5) node{$u$};
   \foreach \x[count=\mycount from 1] in {1.175,1.9,2.55,3.3,4}{
          \draw (\x,0.5) node{$h_a^{\eta}(w_{\mycount})$};
   }
  \end{tikzpicture}
\caption{\label{fig:illustralggeneric}Illustration of systematic resampling in Algorithm \ref{alg:generic}.}
\end{figure*}

Figure \ref{fig:illustralggeneric} illustrates the
systematic resampling used in the first for-loop in Algorithm
\ref{alg:generic}, where $h_a^\eta(w)=w/a$ denotes the expected number
of offspring for a particle with current weight $w < a$. This depends
on the threshold $a$. All particles have $h_a^{\eta}(w_i) <
1$. Particle 1, 2, 4 and 5  each get one descendent and 
particle 3 receives no descendant.

We have considerable freedom in choosing $h_a^{\eta}$ for $w\geq
a$. One natural choice would be
\begin{equation}
\label{eq:hbasic}
  h^\eta_a(w)=
  \begin{cases}
    w/a &  \text{if  } w<a\\
     \lceil w/( \eta a)\rceil  &   \text{if  } w \geq a
  \end{cases}
\end{equation}
illustrated in Figure \ref{fig:ha-funtion}. We call the resulting
algorithm \emph{step-chopthin}. The requirement that
$(w,h^\eta_a(w))\in A$ implies $\lceil w/( \eta a)\rceil \leq \lfloor
w/a \rfloor$ for all $w\geq a$. Considering values of $w$ slightly
less than $2a$ implies $\eta\geq 2$.

This choice does not guarantee the existence of a solution $a$ of
\eqref{eq:defa} due to the discontinuities. Instead of having an exact
solution, one could use an approximate solution, using a numerical
root finding algorithm, but this would not guarantee that the desired
number of particles is returned property (ii).

\begin{figure*}[tb]
  \centering
  \includegraphics[width=0.8\textwidth]{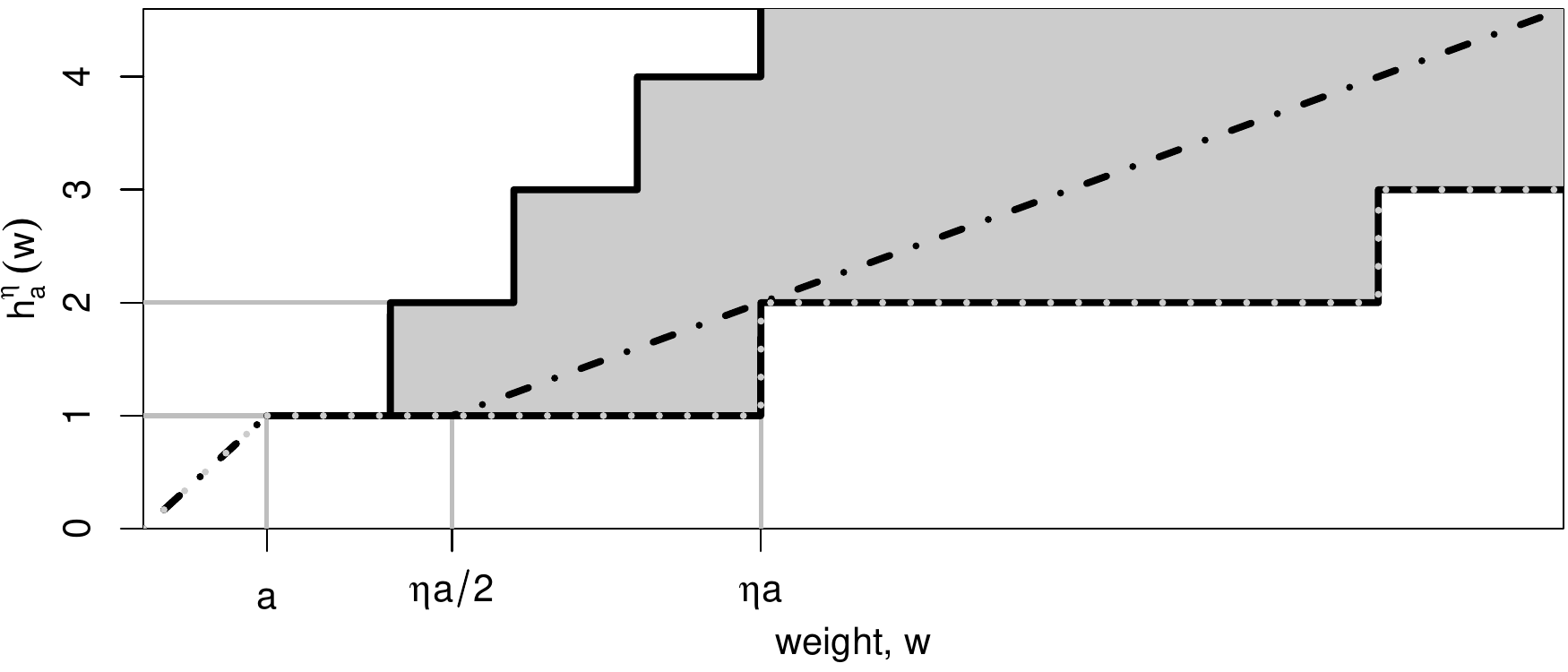}
  \caption{Expected number $h_a^\eta$ of offspring as a function of the weight. Dot-dashed: chopthin (\ref{eq:ha-exact}); gray dotted: step-chopthin (\ref{eq:hbasic}).}\label{fig:ha-funtion}
\end{figure*}

\section{Implementation in expected linear time}
\label{sec:linear}

In this section we present our main version of the algorithm, which we
simply call \emph{chopthin}. For this we choose $h_a^{\eta}$ such that
it is continuous (in $a$) and such that \eqref{eq:defa} can be solved for $a$
in expected linear effort. Consider the function
\begin{equation}\label{eq:ha-exact}
  h_a^\eta(w)=
  \begin{cases}
    w/a & \text{if  } w<a\\
    1 &   \text{if  }a\leq w<\eta a/2\\
    2w/\eta a &  \text{if  } w \geq \eta a/2
  \end{cases}
\end{equation}
which is depicted in Figure \ref{fig:ha-funtion}. The requirement that
$(w,h^\eta_a(w))\in A$ implies $2w/\eta a \leq \lfloor w/a \rfloor$
for all $w\geq a$. Considering values of $w$ slightly less than $2a$
implies $\eta\geq 4$.

\begin{algorithm}[tb]
\caption{Fast determination of $a$}\label{alg:fast}
\DontPrintSemicolon
\KwIn{particle weights $w_i$; maximal weight ratio $\eta$; target
  number of particles $N$}
\KwOut{
 $a>0$ such that 
$\sum_{i=1}^n h_a^\eta(w_i)=N$}
$w^u=w^l=w$, $s^l=0$, $c^m=0$, $s^u=0$, $c^u=0$\;
\While{$w^u\neq \emptyset$ or $w^l\neq \emptyset$}{
  \leIf{$|w^l|\geq | w^u|$}{
    sample $a$ uniformly  from $w^l$ and let
    $b=\eta a /2$\;
    }{
    sample $b$ uniformly  from $w^u$ and let
    $a=2b/\eta$
    }
  $h=s^l/a+\sum_{v\in w^l}\min(v/a,1)+c^m+\sum_{v\in
    w^u}\max(\frac{v}{b}-1,0)+s^u/b-c^u$\;
  \lIf{h=N}{\Return a}
  \eIf{$h>N$}{
    $s^{l}=s^l+\sum_{v\in w^l}v\ind(v\leq a)$\;
    $w^l=\{v\in w^l;v>a\}$, 
    $w^u=\{v\in w^u;v>b\}$\;
  }{
    $c^{m}=c^m+\sum_{v\in w^l}\ind(v\geq a)$, 
    $s^{u}=s^u+\sum_{v\in w^l}v\ind(v\geq b)$, 
    $c^{u}=c^u+\sum_{v\in w^l}\ind(v\geq b)$\;
    $w^l=\{v\in w^l;v<a\}$,
    $w^u=\{v\in w^u;v<b\}$\;
  }
}
\Return $a=\frac{s^l+2 s^u/\eta}{N-c^m+c^{u}}$
\end{algorithm}

We use Algorithm \ref{alg:fast} to solve $\sum_{i=1}^nh_a^\eta(w_i)=N$
using \eqref{eq:ha-exact} for $a$. Lemma \ref{lemma:algo-effort}
proves that the expected effort of Algorithm \ref{alg:fast} is linear
in $n$, and overall the expected computational effort is
$O(\max(n,N))$.

 Algorithm \ref{alg:fast} works by determining which weights are
above or below $a$ and which weights are above or below $\eta a$,
without fully knowing $a$ yet. Due to the piecewise linear structure
of $h_a^{\eta}$, the contributions of weights for which this
determination has been made can be easily kept track of by the number
and the sum of those particles (see $s^l$, $c^m$, $s^u$, $c^u$ and the
computation of $h$ in Algorithm \ref{alg:fast}). The algorithm
maintains two lists --- $w^l$, the weights for which we do not know yet
whether they are above or below $a$, and $w^u$, the weights for which
we do not know yet whether they are above or below $\eta a/2$.  The
exact value of $a$ is only determined when $h=N$ or when both $w^l$
and $w^u$ are empty.  At every iteration, a new candidate for $a$ or
$b=\eta a /2$ is selected from the longer of $w^l$ and $w^u$. Depending
on whether $h>N$ or $h<N$ the algorithm then removes elements from
$w^l$ and $w^u$ and updates the counts/sums of decided weights.
See 
Table \ref{fig:exalgfast} for an illustrative run through of Algorithm
\ref{alg:fast} for $N=n=5$.

\begin{table*}[tb]
      \caption{\label{fig:exalgfast}Example run of Algorithm \ref{alg:fast} with $N=n=5$ and $\eta=4$.}
      \begin{center}
  \begin{tabular}{c|c|c|c|c}
    $w_l$&$a$&$w_u$&$b$&$h$\\\hline

$\{0.1,0.3,0.5,\framebox{0.9},1\}$&0.9&
$\{0.1,0.3,0.5,0.9,1\}$&1.8& 3\\

$\{0.1,0.3,0.5\}$&0.15&
$\{0.1,\framebox{0.3},0.5, 0.9,1\}$&0.3& 9.67\\

$\{0.3,0.5\}$&0.25&
$\{\framebox{0.5},0.9,1\}$&0.5& 6.2\\

$\{\framebox{0.3},0.5\}$&0.3&
$\{0.9,1\}$&0.6& 5.5\\

$\{0.5\}$&0.5&
$\{0.9,\framebox{1}\}$&1& 3.8\\

$\emptyset$&0.45&$\{\framebox{0.9}\}$&0.9&3.89\\
$\emptyset$&&$\emptyset$&&\\
  \end{tabular}\\
$\framebox{\phantom{3} }$ randomly chosen element
\end{center}

\end{table*}

\begin{lemma}\label{lemma:algo-effort}
  The expected effort of Algorithm \ref{alg:fast} is $O(n)$.
  The expected effort of Algorithm \ref{alg:generic} together with Algorithm \ref{alg:fast} is $O(\max(n,N))$.
\end{lemma}
\begin{proof}
  We  use a subscript to denote iterations in Algorithm
  \ref{alg:fast} with $w^l_1=w=w^u_1$. 
The effort in the $i$th iteration of the while-loop is proportional to the number of elements in 
$w^l_i$ and $w^u_i$. Thus the overall effort is proportional to 
$$\sum_{i=1}^{\infty}( |w_i^l|+ | w_i^u |)
$$
Consider iteration $i$. The following  statements are conditional on the sets $w_{i-1}^l,w_{i-1}^u$. 
Suppose that $| w_{i-1}^l| \geq  | w_{i-1}^u|$. We show that
$\E( | w_{i}^l| )\leq (3/4) | w_{i-1}^l| $. Let $a$ be the randomly selected element from $w_{i-1}^l$.
Let $a^{\ast}$ be such that $\sum_{i=1}^n h_{a^{\ast}}^\eta(w_i)=N$.
Let $M^l=\left|\{v\in w_{i-1}^l:v<a^{\ast}\}\right|$,  $M^u=|\{v\in w_{i-1}^l:v>a^{\ast}\}|$. We then have
\begin{align*}
\E&( | w_{i}^l| )=\\
=&\E\left(\! | w_{i}^l| \biggr\vert a<a^{\ast} \!\right)\frac{M^l}{| w_{i-1}^l|}+\E\left(\!| w_{i}^l|\biggr\vert a>a^{\ast}\! \right)\frac{M^u}{| w_{i-1}^l|}\\
=&\left(\frac{M^l}{2}+M^u\right)\frac{M^l}{|w_{i-1}^l|}+\left(\frac{M^u}{2}+M^l\right)\frac{M^u}{|w_{i-1}^l|}\\
=&\frac{(M^l+M^u)^2+2M^lM^u}{2 |w_{i-1}^l|}\leq \frac{1}{2} |w_{i-1}^l|+\frac{1}{4}|w_{i-1}^l|=\frac{3}{4} |w_{i-1}^l|
\end{align*}
Hence, $\E(|w_i^l|+|w_i^u|)\leq \frac{3}{4} |w_{i-1}^l| + |w_{i-1}^u|
\leq \frac{7}{8}(|w_{i-1}^l|+|w_{i-1}^u|)$ as $| w_{i-1}^l| \geq |
w_{i-1}^u|$. Similarly, it can be seen that the above also holds if
$|w_{i-1}^l| < |w_{i-1}^u|$.

Thus, 
\begin{align*}
\E(|w_i^l|&+|w_i^u|)=
\E\left(\E\left[|w_i^l|+|w_i^u| \biggr\vert w_{i-1}^l,w_{i-1}^u\right]\right)\\
&\leq \frac{7}{8}\E(|w_{i-1}^l|+|w_{i-1}^u|)\\
&\leq\dots\leq\left( \frac{7}{8} \right)^{i-1}\E(|w_{1}^l|+|w_{1}^u|)=\left( \frac{7}{8} \right)^{i-1}2n
\end{align*}
and therefore
\begin{align*}
\E &\left[ \sum_{i=1}^{\infty}(|w_i^l|+|w_i^u|) \right]\\
&\leq \sum_{i=1}^{\infty}\left( \frac{7}{8} \right)^{i-1}2n
=2n \frac{1}{1-7/8}=16n.
\end{align*} 
This shows that the expected effort of Algorithm \ref{alg:fast} is
$O(n)$.  The remainder of Algorithm \ref{alg:generic} entails
generating the output of length $N$ and it runs through all $n$
particles with an overall effort of $O(\max(n,N))$. Thus the expected
effort of the combined Algorithms \ref{alg:generic}, \ref{alg:fast} is
$O(\max(n,N))$.\qedhere
\end{proof}

\begin{table}[tb]
  \centering
  \caption{\label{tab:effort}Effort of resampling $N$ particles
    divided by the effort to generate $N$ exponentially distributed random variables in R}
\begin{tabular}{l|p{1pt}rrrr}
  \hline
\hfill $N$&  &$1000$ & $10000$ & $10^5$ & $10^6$ \\ 
  \hline
chopthin && 1.77 & 1.53 & 1.53 & 1.64 \\ 
  systematic && 0.43 & 0.34 & 0.35 & 0.35 \\ 
  multinomial (sample.int) && 0.88 & 0.89 & 1.02 & 1.36 \\ 
  multinomial (cond. Binomial) && 1.81 & 1.90 & 1.92 & 2.03 \\ 
   \hline
\end{tabular}
\end{table}

We now compare the effort of chopthin to the effort of sampling with
replacement (multinomial resampling), via the in-built function
\texttt{sample.int} in R and via a method using conditional Binomial
distributions \citep{DAVIS1993205} and a (fast) C\texttt{++}-based
implementation of systematic resampling.

We simulated $N$ weights from an Exponential distribution,
i.e.  $w_i\sim \text{Exp}(1)$, $i=1,\dots,N$ independently. We then
applied the resampling procedures to the simulated weights. 

Table \ref{tab:effort} reports the mean effort of the resampling
procedures over 10000 repetitions. The reported effort is relative to
the effort to generate the weights (a call of  the
in-built R function rexp). Constant values indicate that the effort is
linear in $N$, as the effort of generating the random variables is
linear in $N$.

Systematic, chopthin and multinomial resampling (the conditional
Binomial implementations) are all approximately linear in $N$. As
expected, chopthin is more computationally demanding than systematic
resampling as part of the chopthin algorithm consists of systematic
resampling steps. 

Nevertheless, the computational effort of chopthin
is very moderate, only slightly more  than generating
exponentially distributed random variables.

\section{Simulations}\label{sec:sim1}
We now compare the performance of chopthin to other resampling methods
within a particle filter. We also vary the bound of the ratio on
the weights, $\eta$, and illustrate that chopthin results in a less
variable ESS.

\subsection{Linear Gaussian Model}\label{sec:linear-model}
Consider a model with hidden Markov process $X_t\in \mathbb{R}$ and observed process $Y_t\in\mathbb{R}$ for $t\in\mathbb{N}$. In this section, we are interested in the model
\begin{equation*}
\begin{cases}
  X_t &= X_{t-1} + \epsilon_t,\quad \epsilon_t \stackrel{\text{iid}}{\sim} N(0,1)\\
  Y_t& = X_t + \xi_t,\quad \xi_t\stackrel{\text{iid}}{\sim} N(0,\sigma^2_Y)
\end{cases}
\end{equation*}
with $X_0\sim N(0,1)$ and known $\sigma_Y>0$. For this model the Kalman filter \citep{kalman1960new} gives  the exact
conditional distribution, giving us a  benchmark.

\subsubsection{Simulation}\label{sec:linear-simulation}

We  use the particle filter in Algorithm \ref{algo:smc}
to give estimates of the hidden states $X_1,\dots,X_T$ based on the 
observations $y_1,\dots,y_T$. We select $p(x_t|x_{t-1})$ and
$p(y_t|x_t)$ as indicated by the linear Gaussian model. We are
interested in the posterior $X_t|y_1,\dots y_t$ for
$t=1,\dots,T$.  Resampling is performed if the ESS drops
below $\beta\in [0,N]$. The ESS of  a weight vector,
$w=(w_1,\dots,w_n)$ is  defined as
\begin{equation*}
  \ESS(w)=\frac{\left(\sum_{i=1}^nw_i \right)^2}{\sum_{i=1}^nw_i^2}.
\end{equation*}
It is  often used in particle filters to trigger the
resampling step. 
 If $\beta = N$ then resampling is performed at every
step as $\ESS( w) \leq N$. Lastly, potentially any resampling scheme
$r$ can be used in Algorithm \ref{algo:smc}.

For a given $\sigma^2_Y$, resampling scheme $r$, target number of
particles $N$ and resampling trigger $\beta$, a single iteration of
the simulation is conducted as follows:\ simulate from the model
$T=1000$ observations; $y_1,\dots,y_T$.  Using this realisation of
observations, run the particle filter to give estimates of the hidden
states $X_1,\dots,X_T$.  Lastly, the Kalman filter is run to obtain
the exact conditional distribution. We use $M=1000$
iterations. The simulation is conducted using combinations of the
parameters: $\sigma_Y$, $N$, $\beta$, $\eta$ (for chopthin) and various
resampling schemes. 

\begin{algorithm}[tb]
\caption{Particle filter}\label{algo:smc}
\DontPrintSemicolon
\KwIn{target number of particles $N$; ESS threshold $\beta$; resampling scheme $r$; observations $y_1,\dots,y_T$.}
\KwOut{ weighted particles $(w_i,\widetilde{x}^{(i)}_{t})_{1:n_i}$ for $t=1,\dots,T$}
Sample $\widetilde{x}_0^{(i)}\sim p(x_0)$, $i=1,\dots,N$ \;
$w_i = 1$, $i=1,\dots,N$ \;
Let $n_0=N$\;
\For{$t=1,\dots,T$}{
Sample $\widetilde{x}_t^{(i)}\sim p(x_t|\widetilde{x}_{t-1}^{(i)})$, $i=1,\dots,n_{t-1}$ \;
$w_i =  w_i p(y_t | \widetilde{x}_t^{(i)})$, $i=1,\dots,n_{t-1}$ \;
\If{$\ESS(w)\leq\beta$}{
Run $r$ with a target of $N$ particles to get a set of particles $(w_i,\widetilde{x}^{(i)}_{t})_{1:n_t}$ \;
Normalise weights such that $\sum_{i=1}^{n_t}w_i=N$
}
}
\end{algorithm}

\newpage
\subsubsection{Illustration of One Run}\label{sec:onerun}

\begin{figure*}[tb]
  \centering
  \includegraphics[width=0.95\linewidth]{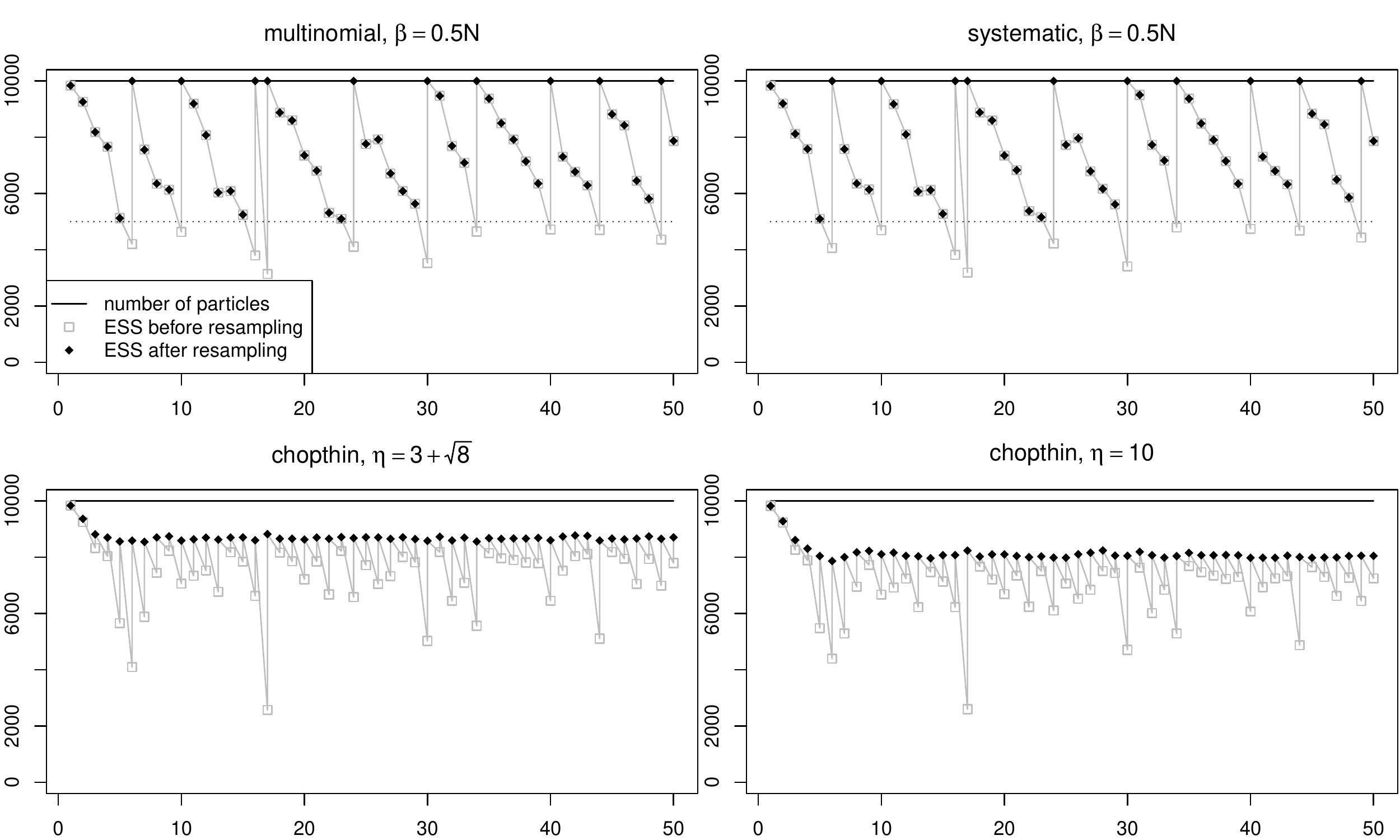}
  \caption{ESS before and after resampling for selected resampling schemes (one realisation).}
  \label{fig:illustresampl}
\end{figure*}

Figure \ref{fig:illustresampl} considers the effect of different
resampling schemes on the ESS during the first 50 steps of one
realisation of the particle filter (Algorithm \ref{algo:smc}) with
$N=10000$ target particles. It plots the ESS before and after
resampling.  As resampling for the multinomial and systematic algorithm
 only occurs if the ESS has dropped below $0.5N$, the ESS is far
more variable than in the chopthin algorithm. For both
$\eta=3+\sqrt{8}$ and $\eta=10$,  the chopthin algorithm
after resampling stays significantly above its theoretical lower bound
(given in Section \ref{sec:controlESS}), which is $0.5N$ and $0.33N$,
respectively.  Also the two choices of  $\eta$ within chopthin
lead to similar behaviour.

\subsubsection{Results}\label{sec:linear-results}
 Table
\ref{tab:MSEs} shows the results of the full simulation for  the following resamplers: chopthin,
multinomial resampling (resampling with replacement), branching
\citep[p.\ 278]{fundamentals}, stratified sampling, standard
residual sampling (multinomial resampling of the residuals), residual
sampling with stratified resampling of the residuals and systematic
resampling.

For each iteration, we obtain the estimated posterior
mean of $X_t$ for $t=1,\dots,T$. For a given $\sigma_Y$, $N$
and $\beta$, denote the estimated posterior mean from iteration $i$,
at time $t$, for resampling scheme $r$ as
$\widetilde{\mu}_{i,t,r}$. Further,  $\mu_{i,t}$ denotes the true posterior mean at
time $t$ given by the Kalman filter. We report the
approximate mean squared error (MSE) for resampling scheme $r$ as
\begin{equation*}
  \frac{1}{M}\sum_{i=1}^{M}\left\{\frac{1}{T}\sum_{t=1}^T(\widetilde{\mu}_{i,t,r} - \mu_{i,t})^2 \right\}.
\end{equation*}
The MSE values, presented in Table \ref{tab:MSEs}, are divided by the
MSE given by the systematic resampling. The results show that using
the chopthin algorithm at every step ($\beta=N$) and using the trigger
($\beta=0.5N$) with various values for the ratio bound $\eta$
consistently achieves a lower MSE than the other resampling
methods. The simulations using $\sigma_Y=1/3$ is based on a setting
where there is a small amount of noise between the state and
observation. In this case, the particle filter will be resampling at
nearly every step for all methods. Chopthin with
$\beta=0.5N$ is included in these simulations to support  our suggestion that
 chopthin should be used in every iteration of a particle filter.
\begin{table*}[tb]
{\centering
\renewcommand{\tabcolsep}{2.5pt}
\caption{Simulations Results - linear Gaussian model: MSE values for various simulation parameters and different resampling methods. Presented MSE values are divided by the MSE from the simulations using systematic resampling.}\label{tab:MSEs}
\begin{tabular}{lll|lllll|llll|llll}
   &  &  & $N$ & 100 & 100 & 100 & 100 & $10^3$ & $10^3$ & $10^3$ & $10^3$ & $10^4$ & $10^4$ & $10^4$ & $10^4$ \\ 
   & $\beta$ & $\eta$ & $\sigma_Y$ & 1/3 & 1 & 3 & 9 & 1/3 & 1 & 3 & 9 & 1/3 & 1 & 3 & 9 \\ 
   \hline
  chopthin & $N$ & 4 &  & {$ { 0.99 }$} & {$ \underline{ 0.90 }$} & {$ \underline{ 0.88 }$} & {$ { 0.91 }$} & {$ { 0.98 }$} & {$ { 0.90 }$} & {$ \underline{ 0.89 }$} & {$ { 0.92 }$} & {$ { 0.97 }$} & {$ { 0.91 }$} & {$ { 0.90 }$} & {$ { 0.94 }$} \\ 
chopthin & $N$ & $3+\sqrt{8}$ &  & {$ { 0.97 }$} & {$ { 0.90 }$} & {$ \underline{ 0.86 }$} & {$ \underline{ 0.86 }$} & {$ { 1.00 }$} & {$ \underline{ 0.89 }$} & {$ \underline{ 0.86 }$} & {$ \underline{ 0.87 }$} & {$ { 0.95 }$} & {$ \underline{ 0.90 }$} & {$ \underline{ 0.89 }$} & {$ { 0.90 }$} \\ 
  chopthin & $N$ & 10 &  & {$ { 0.97 }$} & {$ { 0.92 }$} & {$ \underline{ 0.87 }$} & {$ \underline{ 0.85 }$} & {$ { 0.98 }$} & {$ { 0.91 }$} & {$ \underline{ 0.87 }$} & {$ \underline{ 0.86 }$} & {$ { 0.94 }$} & {$ { 0.93 }$} & {$ \underline{ 0.86 }$} & {$ \underline{ 0.85 }$} \\ 
  chopthin & $0.5N$ & $3+\sqrt{8}$ &  & {$ { 0.98 }$} & {$ { 0.98 }$} & {$ { 0.96 }$} & {$ { 0.94 }$} & {$ { 0.97 }$} & {$ { 0.98 }$} & {$ { 0.96 }$} & {$ { 0.94 }$} & {$ { 0.96 }$} & {$ { 0.98 }$} & {$ { 0.96 }$} & {$ { 0.94 }$} \\ 
  multinomial & $0.5N$ & - &  & {$ { 1.01 }$} & {$ { 1.05 }$} & {$ { 1.15 }$} & {$ { 1.21 }$} & {$ { 0.99 }$} & {$ { 1.04 }$} & {$ { 1.15 }$} & {$ { 1.24 }$} & {$ { 0.98 }$} & {$ { 1.04 }$} & {$ { 1.15 }$} & {$ { 1.22 }$} \\ 
  branching & $0.5N$ & - &  & {$ { 1.01 }$} & {$ { 1.01 }$} & {$ { 0.99 }$} & {$ { 0.99 }$} & {$ { 1.00 }$} & {$ { 0.99 }$} & {$ { 1.00 }$} & {$ { 1.01 }$} & {$ { 0.95 }$} & {$ { 1.01 }$} & {$ { 1.00 }$} & {$ { 1.00 }$} \\ 
  residual & $0.5N$ & - &  & {$ { 1.00 }$} & {$ { 1.00 }$} & {$ { 1.00 }$} & {$ { 0.99 }$} & {$ { 0.99 }$} & {$ { 1.00 }$} & {$ { 1.00 }$} & {$ { 1.01 }$} & {$ { 1.00 }$} & {$ { 1.00 }$} & {$ { 1.00 }$} & {$ { 1.01 }$} \\ 
  stratified & $0.5N$ & - &  & {$ { 1.00 }$} & {$ { 1.00 }$} & {$ { 1.01 }$} & {$ { 1.02 }$} & {$ { 0.99 }$} & {$ { 1.01 }$} & {$ { 1.02 }$} & {$ { 1.04 }$} & {$ { 0.96 }$} & {$ { 1.02 }$} & {$ { 1.01 }$} & {$ { 1.01 }$} \\ 
  residual-stratified & $0.5N$ & - &  & {$ { 1.00 }$} & {$ { 1.01 }$} & {$ { 1.01 }$} & {$ { 1.01 }$} & {$ { 0.97 }$} & {$ { 1.00 }$} & {$ { 1.01 }$} & {$ { 1.01 }$} & {$ { 0.96 }$} & {$ { 1.03 }$} & {$ { 1.00 }$} & {$ { 1.05 }$} \\ 
systematic & $N$ & - &  & {$ { 1.00 }$} & {$ { 0.96 }$} & {$ { 1.06 }$} & {$ { 1.37 }$} & {$ { 1.01 }$} & {$ { 0.96 }$} & {$ { 1.11 }$} & {$ { 1.44 }$} 
& {$ { 0.98 }$} & {$ { 1.00 }$} & {$ { 1.13 }$} & {$ { 1.47 }$} \\
  systematic & $0.5N$ & - &  & {$ { 1.00 }$} & {$ { 1.00 }$} & {$ { 1.00 }$} & {$ { 1.00 }$} & {$ { 1.00 }$} & {$ { 1.00 }$} & {$ { 1.00 }$} & {$ { 1.00 }$} & {$ { 1.00 }$} & {$ { 1.00 }$} & {$ { 1.00 }$} & {$ { 1.00 }$} \\ 
  \end{tabular}
}

\vspace{0.2cm}Underline: below 0.9.
\end{table*}
In similar simulations, not presented here, we compared the MSE of
using the chopthin with $\beta=N$ and systematic resampling with
various values of $\beta$. These simulations still showed that the
chopthin method consistently outperforms systematic resampling.

In general, chopthin appears to perform better than other resampling
methods, particularly when $\sigma_Y$ is large. This may be due to a
combination of factors. First, chopthin with $\beta=N$ keeps the
quality of the particle approximation more stable than methods using
the ESS as resampling trigger (see Figure \ref{fig:illustresampl}).
Second, compared to using a standard resampling scheme at every
iteration ($\beta=N$), chopthin leaves particles with weights between
$a$ and $\eta a/2$ unchanged; only thinning the particles with weights
less than $a$ and chopping those above $\eta a/2$. As a result, a
better particle system seems to be maintained.

\subsubsection{Estimation of the Likelihood}\label{sec:linear-estim-likel}
The likelihood of the observations $p(y_1,\dots,y_T)$ can be decomposed as 
$$
p(y_1,\dots,y_T)=\prod_{t=1}^T  p(y_t|y_{1:t-1}).
$$
The conditional distribution $p(y_t|y_{1:t-1})$ can be approximated
from these simulations the average of the weights; that is
\begin{equation*}
  \widehat{p}(y_t|y_{1:t-1}) = \frac{1}{N}\sum_{k=1}^Nw_k,
\end{equation*}
where the $w_k$ are the weights after the conditioning on the observation $y_t$.

Unbiased estimation of the marginal likelihoods, $p(y_{1:t})$, is
particularly important in particle MCMC methods \cite[e.g.\
][]{particleMCMC:2010,Doucet01062015,sherlock2015} in order to
preserve the correct invariant distribution. We conjecture that the
chopthin algorithm provides an unbiased estimate of the marginal
likelihood. A proof could be based on a decomposition similar to the
one used in \cite[Proposition 7.4.1]{moral2012feynman}.

For the model, presented in Section \ref{sec:linear-model}, the exact
marginal likelihood can be computed using the Kalman filter, providing
a comparison with the estimates given by the particle filter. For the
same run of the simulation conducted in Section \ref{sec:linear-results}, we
estimate the conditional likelihood as follows. Let $y_{1:t}^i$ denote
the observations simulated in iteration $i$ for $t=1,\dots,T$. Then
denote the estimate of $p(y^i_t|y^i_{1:t-1})$ for iteration $i$, for a
given resampling method, $\sigma_Y$, $N$ and $\beta$ as $
\widehat{p}(y^i_t| y^i_{1:t-1})$.  In Table \ref{tab:loglik} we report
the following MSE
\begin{equation*}
  \frac{1}{M}\left\{\sum_{i=1}^M\frac{1}{T}\sum_{t=1}^T\left(\log\widehat{p}( y_t^i| y^i_{1:t-1})-\log p(y^i_t|y^i_{1:t-1}) \right)^2\right\}.
\end{equation*}

\begin{table*}[tb]
{\centering\caption{Simulation Results - linear Gaussian model: MSE values of log likelihood for various simulation parameters and different resampling methods. Presented MSE values are divided by the MSE from the simulations using systematic resampling.}\label{tab:loglik}\renewcommand{\tabcolsep}{2.5pt}
\begin{tabular}{lll|lllll|llll|llll}
   &  &  & $N$ & 100 & 100 & 100 & 100 & $10^3$ & $10^3$ & $10^3$ & $10^3$ & $10^4$ & $10^4$ & $10^4$ & $10^4$ \\ 
   & $\beta$ & $\eta$ & $\sigma_Y$ & 1/3 & 1 & 3 & 9 & 1/3 & 1 & 3 & 9 & 1/3 & 1 & 3 & 9 \\ 
   \hline
chopthin & $N$ & $3+\sqrt{8}$ &  & {$ { 0.92 }$} & {$ \underline{ 0.88 }$} & {$ \underline{ 0.85 }$} & {$ \underline{ 0.86 }$} & {$ { 1.07 }$} & {$ \underline{ 0.88 }$} & {$ \underline{ 0.85 }$} & {$ \underline{ 0.87 }$} & {$ \underline{ 0.89 }$} & {$ { 0.91 }$} & {$ \underline{ 0.89 }$} & {$ \underline{ 0.90 }$} \\ 
  multinomial & $0.5N$ & - &  & {$ { 1.03 }$} & {$ { 1.07 }$} & {$ { 1.17 }$} & {$ { 1.22 }$} & {$ { 0.94 }$} & {$ { 1.07 }$} & {$ { 1.17 }$} & {$ { 1.24 }$} & {$ { 1.09 }$} & {$ { 1.06 }$} & {$ { 1.17 }$} & {$ { 1.23 }$} \\ 
  branching & $0.5N$ & - &  & {$ { 1.02 }$} & {$ { 1.00 }$} & {$ { 0.99 }$} & {$ { 0.99 }$} & {$ { 0.98 }$} & {$ { 0.99 }$} & {$ { 1.00 }$} & {$ { 1.01 }$} & {$ { 1.03 }$} & {$ { 1.02 }$} & {$ { 1.00 }$} & {$ { 1.01 }$} \\ 
  residual & $0.5N$ & - &  & {$ { 1.00 }$} & {$ { 1.00 }$} & {$ { 1.00 }$} & {$ { 1.00 }$} & {$ { 1.01 }$} & {$ { 0.99 }$} & {$ { 1.00 }$} & {$ { 1.01 }$} & {$ { 1.00 }$} & {$ { 1.00 }$} & {$ { 1.00 }$} & {$ { 1.02 }$} \\ 
  stratified & $0.5N$ & - &  & {$ { 1.03 }$} & {$ { 0.99 }$} & {$ { 1.02 }$} & {$ { 1.03 }$} & {$ { 0.95 }$} & {$ { 1.01 }$} & {$ { 1.02 }$} & {$ { 1.05 }$} & {$ { 0.91 }$} & {$ { 1.03 }$} & {$ { 1.02 }$} & {$ { 1.01 }$} \\ 
  systematic & $N$ & - &  & {$ { 0.99 }$} & {$ { 0.95 }$} & {$ { 1.06 }$} & {$ { 1.37 }$} & {$ { 1.05 }$} & {$ { 0.93 }$} & {$ { 1.10 }$} & {$ { 1.45 }$} & {$ { 0.95 }$} & {$ { 0.99 }$} & {$ { 1.11 }$} & {$ { 1.46 }$} \\ 
  systematic & $0.5N$ & - &  & {$ { 1.00 }$} & {$ { 1.00 }$} & {$ { 1.00 }$} & {$ { 1.00 }$} & {$ { 1.00 }$} & {$ { 1.00 }$} & {$ { 1.00 }$} & {$ { 1.00 }$} & {$ { 1.00 }$} & {$ { 1.00 }$} & {$ { 1.00 }$} & {$ { 1.00 }$} \\ 
  \end{tabular}
}
\end{table*}
Based on the MSE results, the chopthin method approximates the log
likelihood better than systematic and consistently for other
resampling methods.

\subsection{Stochastic Volatility Model}\label{sec:stochvol-model}
We now consider a more complicated model; a stochastic volatility
model with hidden process $X_t\in\mathbb{N}$ and $Y_t\in\mathbb{R}$
for $t\in\mathbb{N}$ defined by:
\begin{equation*}
  \begin{cases}
    X_t &= 0.9 X_{t-1}+0.25 \epsilon_t,\quad \epsilon_t\stackrel{\text{iid}}{\sim} N(0,1)\\
    Y_t&=0.1 \xi_t\exp(X_t/2),\quad \xi_t\stackrel{\text{iid}}{\sim} N(0,1)
  \end{cases}
\end{equation*}
with $X_0\sim N(0,1)$. Unlike the linear Gaussian model, the posterior
distributions are not available in closed form. As a benchmark we approximate these
distributions using a numerical approach that discretises the hidden
state space into a fine grid. 

\subsubsection{Results}\label{sec:stochvol-results}
We repeat the same simulation described in Section
\ref{sec:linear-simulation} for the stochastic volatility
model. altering Algorithm \ref{algo:smc} accordingly. The MSE of the
posterior mean and loglikelihood is presented in Table
\ref{tab:stochvol-MSEs+ll}. The results again show that using chopthin
every iteration outperforms the other resampling method. 

\begin{table*}\centering\renewcommand{\tabcolsep}{3.5pt}
  \caption{Simulations Results - Stochastic volatility model: MSE values for different resampling methods. Presented MSE values are divided by the MSE from the simulations using systematic resampling. Left table:\ MSE of the posterior mean, right table:\ MSE of the loglikelihood.}\label{tab:stochvol-MSEs+ll}
\begin{minipage}{0.48\textwidth}
\centering
\begin{tabular}{llll|lll|}
   & $\beta$ & $\eta$ &  & $100$ & $10^3$ & $10^4$ \\ 
   \hline
  chopthin & $N$ & 4 &  & {$ \underline{ 0.82 }$} & {$ \underline{ 0.82 }$} & {$ \underline{ 0.85 }$} \\ 
   chopthin & $N$ & $3+\sqrt{8}$ &  & {$ \underline{ 0.84 }$} & {$ \underline{ 0.83 }$} & {$ \underline{ 0.87 }$} \\ 
  chopthin & $N$ & 10 &  & {$ \underline{ 0.89 }$} & {$ \underline{ 0.89 }$} & {$ \underline{ 0.90 }$} \\ 
  chopthin & $0.5N$ & $3+\sqrt{8}$ &  & {$ { 1.05 }$} & {$ { 1.04 }$} & {$ { 1.01 }$} \\ 
  multinomial & $0.5N$ & - &  & {$ { 1.09 }$} & {$ { 1.10 }$} & {$ { 1.05 }$} \\ 
  branching & $0.5N$ & - &  & {$ { 1.00 }$} & {$ { 1.00 }$} & {$ { 0.99 }$} \\ 
  residual & $0.5N$ & - &  & {$ { 1.01 }$} & {$ { 1.00 }$} & {$ { 1.00 }$} \\ 
  stratified & $0.5N$ & - &  & {$ { 1.00 }$} & {$ { 1.00 }$} & {$ { 0.99 }$} \\ 
  systematic & $N$ & - &  & {$ { 1.00 }$} & {$ { 1.00 }$} & {$ { 0.98 }$} \\ 
  systematic & $0.5N$ & - &  & {$ { 1.00 }$} & {$ { 1.00 }$} & {$ { 1.00 }$} \\ 
  \end{tabular}
\end{minipage}%
\hfill
\begin{minipage}{0.48\textwidth}
\centering
\begin{tabular}{llll|lll|}
   & $\beta$ & $\eta$ &  & $100$ & $10^3$ & $10^4$ \\ 
   \hline
  chopthin & $N$ & $3+\sqrt{8}$ &  & {$ \underline{ 0.85 }$} & {$ \underline{ 0.83 }$} & {$ \underline{ 0.88 }$} \\ 
  multinomial & $0.5N$ & - &  & {$ { 1.08 }$} & {$ { 1.09 }$} & {$ { 1.05 }$} \\ 
  branching & $0.5N$ & - &  & {$ { 0.97 }$} & {$ { 1.00 }$} & {$ { 0.99 }$} \\ 
  residual & $0.5N$ & - &  & {$ { 0.99 }$} & {$ { 1.01 }$} & {$ { 0.97 }$} \\ 
  stratified & $0.5N$ & - &  & {$ { 1.00 }$} & {$ { 1.00 }$} & {$ { 0.98 }$} \\ 
  systematic & $0.5N$ & - &  & {$ { 1.00 }$} & {$ { 1.00 }$} & {$ { 1.00 }$} \\ 
  \end{tabular}
\end{minipage}
\end{table*}

The results in Tables \ref{tab:MSEs}, \ref{tab:loglik} and
\ref{tab:stochvol-MSEs+ll} show that for a fixed number of particles,
chopthin outperforms other resamplers in terms of MSE. However, as
illustrated in Table \ref{tab:effort}, using chopthin is
computationally more expensive than systematic resampling. Therefore,
use of chopthin should be favoured when the computational expense of
the other steps in the particle filter, i.e.\ the transition of the
particle values and computational of the weights, exceed the expense
of resampling. In scenarios where the transition or weight computation
are cheap, using systematic resampling may be preferred.

\newpage
\section{Implied control of the Effective Sample Size}
\label{sec:controlESS}
The following lemma shows that imposing a bound on the ratio between
the weights implicitly results in a lower bound on the ESS. It implies
that chopthin has a lower bound on the ESS after resampling.

\begin{lemma}
Suppose $w_1,\dots,w_n>0$.
Then
$$
\ESS( w)=\frac{(\sum_{i=1}^n w_i)^2}{\sum_{i=1}^n w_i^2}\geq  4\frac{\eta n+1-\eta^2}{(\eta+1)^2}
$$
where $\eta=\frac{\max_i w_i}{\min_i w_i}$.
\end{lemma}

\begin{proof}
In the case where all weights are equal, i.e.\ $\eta=1$, then $\ESS(w)=n$, thus inequality holds.
From now on consider the case $\eta>1$. 

Let $W_i=\frac{w_i}{\sum_jw_j}$ be the normalized weights
corresponding to $W$.  Then $\ESS(w)=\ESS(W)$. The set of
possible normalized weights is compact and $\ESS$ is a continuous
function, thus there exists a $W^{\ast}$ that minimises $\ESS$.
Without loss of generality, assume $W^{\ast}_1\leq \dots \leq
W^{\ast}_n$.

The normalised weight $W^{\ast}_i$ has to be of the form
$W^{\ast}_i=a$ for $i<k$, $W^{\ast}_k=a\tau$,
$W^{\ast}_i=\eta a$ for $i<k$, where $k\in
\{1,\dots,n-1\}$, $a>0$ and $1\leq \tau<\eta$.  To see this let $W$ be
a normalised weight vector for which there exist mutually distinct
indices $i,j,k,l$ such that $W_i<W_j\leq W_k<W_l$. Define a new weight vector $V$ identical
to $W$ except for $V_j=W_j-\Delta$, $V_{k}=W_k+\Delta$ with
$\Delta=\min((W_l-W_k)/2, (W_j-W_i)/2)$.  Then
\begin{align*}
  1/&\ESS(V)=\sum_\nu V_\nu^2\\
  &=2\Delta^2+2\Delta(W_k-W_j)
  +\sum_{\nu}W_\nu^2>1/\ESS(W)
\end{align*}
which shows that $W$ does not minimise $\ESS$. 
Hence, $W^{\ast}$ can take at most 3 values, the middle one, if present, appearing exactly once.
The two extreme values have to have a ratio of $\eta$, otherwise one could move them further apart and create a weight vector with smaller $\ESS$.

As
$\sum W_i^{\ast}=a[k-1+\tau+\eta(n-k)]$, we have
\begin{align*}
\ESS(W^\ast)=&\frac{[k-1+\tau+\eta(n-k)]^2}{k-1+\tau^2+(n-k)\eta^2}\\
&\geq \frac{[k+\eta(n-k)]^2}{k-1+(n-k+1)\eta^2}
\geq \inf_{x\in [1,n-1]} h(x)
\end{align*}
where 
$h(x)=\frac{[x+\eta(n-x)]^2}{x-1+(n-x+1)\eta^2}$.

It remains to derive the minimum of $h$.
Candidates for minimizers of $h$ are $x=\eta n/(\eta-1)$ (which is not in the
right range) and $x = (\eta (n+2)+2)/(\eta+1)$. 
Plugging this into $h$ gives
$h(x)\geq 4\frac{\eta n+1-\eta^2}{(\eta+1)^2}$
\end{proof}

Larger $\eta$ allow for more variability in the weights and thus should lead to lower effective sample sizes. Consistent with this, the lower bound on $\ESS$ is decreasing in $\eta$. This can be seen by differentiating it with respect to $\eta$.

For large $n$, the leading term is
$4\frac{\eta n}{(\eta+1)^2}$. Equating this to a desired minimal effective sample size $\gamma n$ gives
$$
\eta = \frac{2-\gamma+2 \sqrt{1-\gamma}}{\gamma}
$$
For example, for $\gamma=0.5$, this leads to
$\eta=3+\sqrt{8}$. Furthermore, for $\eta=10$, the lower bound on the
$\ESS$ is $\frac{40}{121}n-\frac{99}{121}\approx 0.33n$.

\section{Discussion}
\subsection{Why not only impose an upper or a lower threshold on the weights?}

The chopthin algorithm imposes a bound on the ratio of the largest
and the smallest weight. Alternatively, one could have imposed only a
lower or only an upper bound on the normalized weights. The following
examples illustrate that there are situations in which these bounds
would not lead to resampling despite very uneven weights. The
chopthin algorithm (with $\eta<n$) would even out the weights in both
examples.

\begin{example}
  Suppose our weight vector $w$ of length $n$ is produced by one
  importance sampling step, where the target distribution is a uniform
  distribution on $[0,0.5]$ and the importance sampling distribution
  is a uniform distribution on $[0,1]$. Then roughly half of the
  weights will be approximately $2/n$ and half of the weight will be
  0. None of the weights is large, so imposing
  an upper bound on the weights would not lead to resampling.
\end{example}

\begin{example}
  Consider the same setting as in the previous example, but now having
  as target distribution a mixture of two equally probably components:
  a uniform distribution on $[0,1]$ and a uniform distribution on
  $[0,1/n]$.  Suppose the first particle is in [0,1/n] and all other
  samples are greater than $1/n$.  Then the weight of the first
  particle is $(n+1)/2$ and the weight of all other particles is
  $1/2$.  Thus in this case, no weight is small, so imposing a lower
  bound on the weights would not lead to resampling.
\end{example}

\section{Summary}
In this paper we have introduced the chopthin algorithm which bounds
the ratio between the weights. We showed, in simulations, that
chopthin consistently outperforms standard resampling schemes used in
particle filters. The simulations also demonstrated that chopthin can
be used at every iteration in a particle filter with no detrimental
effects. The chopthin algorithm can be implemented efficiently and we
have proved that its expected effort is linear in the number of
samples. Lastly, we have shown that imposing a bound on the ratio
between weights implicitly controls the ESS. As mentioned in Section
\ref{sec:sim1}, use of chopthin within particle filters over other,
less computational expensive, resamplers should be favoured when the
expense of resampling is negligible in comparison to the other steps
in the particle filter.

Proving a central limit type theorem of the particle filter estimates
using chopthin resampling is a natural next step. However, as the
chopthin algorithm uses systematic resampling this will not be
straightforward \citep{gentil2008}. Replacing systematic resampling
with a resampling method more amenable to theoretical developments
could be a topic for future research.

\end{document}